\documentclass[12pt]{iopart}

\usepackage{amssymb,amsthm,graphicx}
\usepackage{graphicx}
\usepackage{colordvi}

\newcommand{\R}{\mathbb{R}}

\newcommand{\OO}{\mathcal{O}}
\newcommand{\D}{\mathrm{d}}

\newtheorem{claim}{Claim}[section]
\newtheorem{thm}[claim]{Theorem}
\newtheorem{prop}[claim]{Proposition}
\newtheorem{lem}[claim]{Lemma}
\newtheorem{rems}[claim]{Remarks}

\begin{document}

\title[Approximations of vertex couplings by singularly scaled potentials]
{Approximations of quantum-graph vertex couplings by singularly scaled potentials}

\author{Pavel Exner}
\address{Doppler Institute for Mathematical Physics and Applied
Mathematics, \\ Czech Technical University in Prague,
B\v{r}ehov\'{a} 7, 11519 Prague, \\ and  Nuclear Physics Institute
ASCR, 25068 \v{R}e\v{z} near Prague, Czechia} \ead{exner@ujf.cas.cz}

\author{Stepan S. Manko}
\address{Department of Physics, Faculty of Nuclear Science and Physical Engineering, \\ Czech Technical University in Prague,
Pohrani\v{c}n\'{\i} 1288/1,  40501 D\v{e}\v{c}\'{\i}n,
Czechia\footnote{On leave of absence from Pidstryhach Institute for Applied Problems of Mechanics and Mathematics, National Academy of Sciences of Ukraine, 3b Naukova str, 79060 Lviv, Ukraine}} \ead{stepan.manko@gmail.com}

\begin{abstract}
We investigate the limit properties of a family of Schr\"odinger operators of the form $H_\varepsilon= -\frac{\mathrm{d}^2}{\mathrm{d}x^2}+ \frac{\lambda(\varepsilon)}{\varepsilon^2}Q \big(\frac{x}{\varepsilon}\big)$
acting on $n$-edge star graphs with Kirchhoff conditions imposed at the vertex.
The real-valued potential $Q$ is supposed to have compact support and $\lambda(\cdot)$ to be analytic around $\varepsilon=0$ with $\lambda(0)=1$. We show that if the operator has a zero-energy resonance of order $m$ for $\varepsilon=1$ and $\lambda(1)=1$, in the limit $\varepsilon\to 0$ one obtains the Laplacian with a  vertex coupling depending on $1+\frac12 m(2n-m+1)$ parameters. We prove the norm-resolvent convergence as well as the convergence of the corresponding on-shell scattering matrices. The obtained vertex couplings are of scale-invariant type provided $\lambda'(0)=0$; otherwise the scattering matrix depends on energy and the scaled potential becomes asymptotically opaque in the low-energy limit.
\end{abstract}

\maketitle

\section{Introduction}

Quantum graphs attracted a lot of attention since their rediscovery in the second half of the 1980's; we refer to the recent monograph \cite{BeKu13} for a broad overview and an extensive bibliography. One of the central questions in these models concerns the way in which the wave functions are coupled at the graph vertices. If $n$ edges meet at a vertex, in the absence of external fields the requirement of probability current conservation leads to the condition
\[
(U-I)\Psi(0) + \mathrm{i}(U+I)\Psi'(0) = 0
\]
coupling the vectors of boundary values of the wave functions and their derivatives,
in which $U$ is an $n\times n$ unitary matrix. This tells us, in particular, that such a coupling may depend on $n^2$ real parameters.

Different $U$ give rise to different dynamics on the graph and the choice of $U$ should be guided by the physical contents of the model one is constructing, in the first place by properties of the junctions to which the graph vertices should represent an idealized description. A natural approach to this problem is to start from the most simple coupling, often called Kirchhoff, and to investigate how the junction properties are influenced by a potential supported in the vicinity of the vertex, in particular if the support shrinks to a point and the potential is properly scaled. It is easy to obtain in this way the so-called $\delta$ coupling using the scaling which preserves mean value of the potential \cite{Ex96b}. However, this is just a one-parameter subset of all the admissible matching conditions, and the other ones require a different limiting procedure, for instance, using shrinking potentials with a more singular scaling of the type $Q(\cdot) \mapsto \varepsilon^{-2}\,Q \big(\frac{\cdot} {\varepsilon}\big)$.

Such limits were investigated first for vertices connecting two edges, which is equivalent to generalized point interaction on the line \cite{AGHH05}, with the conclusion that the limit is trivial describing disconnected edges \cite{Seb86}. Later it was pointed out, however, that such a claim holds only generically and a nontrivial limit may exist when the potential $Q$ has a zero-energy resonance -- cf.~\cite{CAZEG03} and subsequent papers of these authors, see also \cite{GolMan09, GolHry10, GolHry13}. One has to stress, however, that the role of zero-energy resonances in the limit was in fact known before; one can find it in the analysis of the one-dimensional low-energy scattering \cite{BoGeWi85}. The result was further generalized, in particular, to Schr\"odinger operators on star graphs \cite{Man10, Man12} or to combinations of potentials with different scaling \cite{Gol13}.

One should mention that an inattentive reader may run into a terminological confusion in this area coming from the fact that the question addressed in the original \v{S}eba's paper \cite{Seb86} concerned the possibility of approximating the $\delta'$ interaction. This name was given thirty years ago, maybe not quite fortunately, to a class of point interactions characterized by an effective Neumann decoupling in the high-energy limit \cite[Sect.~I.5]{AGHH05}. Although the answer found in \cite{Seb86} was negative and the omitted non-generic limits found later described a different type of point interactions, they are nevertheless sometimes labeled as $\delta'$. One of the conclusions of this paper is that similar limits on star graphs lead to couplings which are \emph{not} of the $\delta'$ type in the conventional sense \cite{AGHH05, Ex96}.

The present study can be regarded as an extension of the previous work of one of us \cite{Man10, Man12} where such limits were studied on three-legged star graphs. Those results are generalized here in several ways. First of all, we consider a star graph of $n$ edges with an arbitrary finite $n$. Equally important, we consider potential families of the type $\frac{\lambda(\varepsilon)}{\varepsilon^2}Q \big(\frac{x}{\varepsilon}\big)$ with an $\varepsilon$-dependent coupling parameter, which allows us to obtain a family of limiting couplings which will contain not only scale-invariant matching conditions --- this may happen provided $\lambda'(0)\ne 0$ --- including graph Hamiltonians with a nonempty discrete spectrum. We note that in the case $n=2$ this conclusion reduces to a particular case of the two-scale limit result of Golovaty \cite{Gol12, Gol13}.

We shall establish the norm-resolvent convergence of the scaled operator families as well as the convergence of the corresponding on-shell scattering matrices, in contrast to \cite{Man10},
where the part of the scattering matrix was studied. In the next section we formulate the problem rigorously and state our main results in Theorems~\ref{thm:operator} and \ref{thm:Scat}; the following two sections are devoted to their proofs and discussion.

\section{Preliminaries and main results}

We begin with recalling a few basic notions from the theory of differential equations on graphs. A metric graph $G =(V,E)$ is identified with finite sets $V=V(G)$ of \emph{vertices} and $E=E(G)$ of \emph{edges}, the latter being isomorphic to (finite or semi-infinite) segments of the real line. One can think of the graphs as being embedded into the Euclidean space, with the vertices being points of $\R^3$ and edges smooth
regular curves connecting them, but such an assumption plays no role in the following.

A map $f:G\to\mathbb{C}$ is said to be a \emph{function on the graph} and its restriction to the edge $e\in E(G)$ will be denoted by $f_e$.
Each edge has a natural parametrization; if $G$ is embedded into $\mathbb{R}^3$ it is given by the arc length of the curve representing the edge.
A differentiation is always related to this natural length parameter.
Vertices are endpoints of the corresponding edges; we denote by $\frac{\D f}{\D e}(a)$ the limit value of the derivative at the point $a\in V(G)$ taken conventionally in the outward direction, i.e. away from the vertex. The integral $\int_G f\,\D s$ of $f$ over $G$ is the sum of integrals over all edges, the measure being the natural Lebesgue measure. Using this notion we can introduce the Hilbert space $L^2(G)$ with the scalar product $(f,h)=\int_Gf\bar h\,\D s$, and furthermore, the Sobolev space $H^2(G)$ on the graph with the norm
\[
\|f\|_{H^2(G)}=\big(\|f\|_{L^2(G)}+\|f''\|_{L^2(G)}\big)^{1/2}.
\]
Observe that neither the function belonging to $H^2(G)$ nor its derivative should be continuous at the graph vertices.
In what follows, the symbol $BC(G)$ will stand for the Banach space of functions that are continuous and bounded on each edge $e\in E(G)$ with the supremum norm.
In a similar way, we introduce the Banach space $BC^1(G)$ consisting of functions that are continuous and bounded on each edge $e\in E(G)$ along with their first derivatives.
If $f\in BC^1(G)$, then the symbol $\|f\|_{BC^1(G)}$ stands for the sum of its supremum and that of its first derivative.
We shall also need the space $C^\infty(G)$ of functions infinitely differentiable on each edge of the graph $G$, and finally, $L^\infty(G)$ will stand for the standard Lebesgue-measurable function space on the graph with the essential-supremum norm.

We say that a function $f$ satisfies the Kirchhoff conditions at the vertex $a\in V(G)$ if $f$ is continuous at this vertex and $\sum_e \frac{\D f}{\D e}(a) = 0$ holds, where the sum is taken over all the edges incident in $a$; in the particular case when there is only one such edge $e$ the Kirchhoff conditions at the ``hanging'' vertex $a$ reduce to
the usual Neumann condition, $\frac{\D f}{\D e}(a)=0$.
The symbol $K(G)$ shall denote the set of functions on $G$ obeying the Kirchhoff conditions at each graph vertex.

Since our approximation problem is of a local character we focus on noncompact star-shaped graphs $\Gamma$ consisting of $n$ semi-infinite edges $\gamma_1,\dots, \gamma_n$ connected at a single vertex denoted by $a$. In that case $E(\Gamma)=\{\gamma_i\}_{i=1}^n$ and without loss of generality we may identify each $\gamma_i$ with the halfline $[0,\infty)$. Our consideration will need neighborhoods of the vertex; if $a_i$ stands for an arbitrary but fixed point of $\gamma_i$,
we denote by $\omega_i$ part of the edge $\gamma_i$ connecting the root vertex $a$ and the point $a_i$ and then introduce the compact star graph $\Omega$ with edges $E(\Omega)=\{\omega_i\}_{i=1}^n$ and vertices $V(\Omega)=\{a\}\cup\{a_i\}_{i=1}^n$.
With the chosen parametrization in mind, we use the symbol $a_i$ for both the vertex and its distance from $a$.

Given a star graph $\Gamma$, we introduce the following family of Schr\"{o}dinger operators on $L^2(\Gamma)$ labeled by the parameter $\varepsilon\in(0,1]$,
\[
H_\varepsilon=-\frac{\mathrm{d}^2}{\mathrm{d}x^2}+
\frac{\lambda(\varepsilon)}{\varepsilon^2}Q
\Big(\frac{x}{\varepsilon}\Big),
\qquad \mathrm{dom}\,H_\varepsilon=H^2(\Gamma)
\cap K(\Gamma),
\]
where the real-valued potential $Q$ belongs to the class $L^\infty(\Gamma)$ and has a compact support supposed to be a subset of the graph $\Omega$ constructed above. In fact, we identify $\Omega$ with $\mathrm{supp}\,Q$ unless the intersection of the support with some $\gamma_i$ is the vertex $a$ only, in which case we choose $a_i>0$ on such edges. With respect to the edge indices $Q$ may be regarded as an $n\times n$ matrix function on $[0,\infty)$; we stress that it need not be diagonal. In a similar vein the differential part of $H_\varepsilon$ is a shorthand for the operator which acts as the negative second derivative on each edge $\gamma_i$. The function $\lambda(\cdot)$ in the above expression is supposed to be
real-valued for real $\varepsilon$ and
holomorphic in the vicinity of the origin. In addition, it  satisfies the condition
\[
\lambda(\varepsilon)=1+\varepsilon\lambda+\OO(\varepsilon^2),
\quad \varepsilon\to0,
\]
where $\lambda$ is a real number. Our main goal in this paper is to investigate convergence of the operators $H_\varepsilon$ as $\varepsilon\to 0$ in the norm-resolvent topology.

To state the results, we need a few more notions. First of all, we denote by $\Omega_\varepsilon$ the $\varepsilon$-homothety of the graph $\Omega$, centered at $a$, i.e. the subgraph of $\Gamma$ with the vertices $V(\Omega_\varepsilon):= \{a\}\cup\{a^\varepsilon_i\}_{i=1}^n$, where $a^\varepsilon_i= \varepsilon^{-1} a_i$.
In particular, if $\Omega= \mathrm{supp}\,Q$, then $\Omega_\varepsilon$ is the support of $Q(\varepsilon^{-1}\cdot)$.

Furthermore, we shall say that a Schr\"{o}dinger operator in ${L}^2(\Gamma)$ of the form
\[
S:=-\frac{\mathrm{d}^2}{\mathrm{d}x^2}+Q
\]
satisfying Kirchhoff conditions has
a \emph{zero-energy resonance of order} $m$
if there exist $m$ linearly independent \emph{resonant solutions} $\psi_1,\dots,\psi_m$ to the equation
\begin{equation}\label{PreResonant}
-\psi''+Q\psi=0
\end{equation}
which are bounded on $\Gamma$. As in the particular case considered earlier in \cite{Man12}, the limit behavior of $H_\varepsilon$ will depend crucially on the existence and properties of a zero-energy resonance of the operator $S$. Since every bounded solution of the equation (\ref{PreResonant}) is constant outside the support of $Q$, it follows that $\psi_i$ solves the Neumann problem
\begin{eqnarray} 
\label{Resonant}
-\psi''+Q\psi=0
\quad\mathrm{on}\quad\Omega,\qquad
\psi\in{K}(\Omega),
\end{eqnarray} 
hence the existence of a zero-energy resonance can (and shall) be reformulated in terms of the problem (\ref{Resonant}):
we say that the Schr\"{o}dinger operator $S$ has a zero-energy resonance of order $m$ if there are $m$ linearly independent (resonant) solutions to the problem (\ref{Resonant}).

If $m>1$, the construction we are going to present below requires a particular basis in the space of solutions of the problem (\ref{Resonant}). Consider an arbitrary pair $\varphi_1,\,\varphi_2$ of linearly independent solutions; then there must exist two different vertices $b_1$ and $b_2$ from the set $\{a_i\}_{i=1}^n$ such that $\varphi_1(b_1)\neq0$ and $\varphi_2(b_2)\neq0$. If the quantity $\phi:=\varphi_1(b_1) \varphi_2(b_2) -\varphi_2(b_1)\varphi_1(b_2)$ is nonzero, we can define the functions
\[
\psi_1:=
\big(\varphi_2(b_2)\varphi_1-\varphi_1(b_2)\varphi_2\big)
/\phi,
\qquad
\psi_2:=
\big(\varphi_1(b_1)\varphi_2-\varphi_2(b_1)\varphi_1\big)
/\phi.
\]
Furthermore, we can renumber the edges in such a way that $a_1:=b_1$ and $a_2:=b_2$, then a short computation shows that
\begin{equation}\label{psi11}
\psi_1(a_1)=\psi_2(a_2)=1,\qquad \psi_1(a_2)=\psi_2(a_1)=0.
\end{equation}
If, on the other hand, $\phi$ is zero, we can find a vertex $b_3\in \{a_i\}_{i=1}^n$ at which the function  $\varphi_2(b_2)\varphi_1 -\varphi_1(b_2)\varphi_2$ does not vanish and set
\begin{eqnarray*}
\psi_2:=
\big(\varphi_2(b_2)\varphi_1-\varphi_1(b_2)\varphi_2\big)/
\big(\varphi_2(b_2)\varphi_1(b_3)-\varphi_1(b_2)\varphi_2(b_3)\big),
\\
\psi_1:=\big(\varphi_1-\varphi_1(b_3)\psi_2\big)/
\varphi_1(b_1);
\end{eqnarray*}
rearranging the edges in such a way that $a_1:=b_1$ and $a_2:=b_3$ we get again (\ref{psi11}). The process can be continued leading to the following conclusion:

\begin{lem}\label{lem:eigenfunction}
Suppose that the problem (\ref{Resonant}) has $m$ linearly independent solutions, then one can choose them as real-valued functions $\psi_1,\dots,\psi_m$ satisfying
\[
\psi_i(a_j)=\delta_{ij},\quad i,j=1,\dots,m,
\]
where $\delta_{ij}$ is the Kronecker symbol.
\end{lem}
\begin{proof}
It remains to check that the $\psi_i$'s can be chosen real which follows from the fact that the operator $S$ including the Kirchhoff conditions at the origin commutes with the complex conjugation.
\end{proof}

For notational convenience, we introduce the resonant and full index sets by
\[
\mathfrak{m}:=\{1,\dots,m\},\qquad
\mathfrak{n}:=\{1,\dots,n\},
\]
respectively, adopting the convention that $\mathfrak{m}$ is empty for $m=0$. To describe the outcome of the limiting process we need the following quantities:
\begin{eqnarray*}
\theta_{{i}{j}}:=\psi_{i}(a_{j}),\qquad {i}\in{\mathfrak{m}},\quad {j}\in\mathfrak{n}\setminus\mathfrak{m},
\\
q_{ij}:=\int_\Gamma Q\psi_i\psi_j\,\D \Gamma,\qquad i,j\in\mathfrak{m}.
\end{eqnarray*}
Using them, we define the {\it limit operator} $H$ as the one acting via
\[
H\phi:=-\phi''
\]
on functions $\phi\in{H}^2(\Gamma)$ that obey the matching conditions
\begin{eqnarray} 
\phi_{\gamma_{j}}(a)-\sum_{{i}\in\mathfrak{m}}\theta_{{i}{j}}
\phi_{\gamma_{i}}(a)=0,
\quad {j}\in\mathfrak{n}\setminus\mathfrak{m}, \nonumber\\[-.5em] \label{limit:mc} \\[-.5em]
\frac{\mathrm{d}\phi}{\mathrm{d}\gamma_{i}}(a)+
\sum_{{j}\in\mathfrak{n}\setminus\mathfrak{m}}\theta_{{i}{j}}\frac{\mathrm{d}\phi}
{\mathrm{d}\gamma_{j}}(a)-
\lambda
\sum_{j\in\mathfrak{m}}{q}_{ij}\phi_{\gamma_j}(a)
=0,
\quad {i}\in\mathfrak{m}. \nonumber
\end{eqnarray} 

\begin{rems} \label{bc-rem}
{\rm (a) In the family of the operators with potentials $Q$ of the described class a zero-energy resonance is generically absent. This situation corresponds to Dirichlet decoupled edges, $\phi_{\gamma_j}(a)=0$.

\smallskip

\noindent (b) If $\lambda=0$ or $Q=0$, the conditions (\ref{limit:mc}) do not couple function values and derivatives, and as a result, the matching conditions of the limit operator are \emph{scale-invariant}. This means, in particular, that $H$ has no eigenvalues and $\sigma(H)=[0,\infty)$. Another manifestation of the scale-invariant character is that the scattering matrix, which we shall discuss below, is independent of energy.

\smallskip

\noindent (c) The matching conditions (\ref{limit:mc}) contain in general
$m(n-m)+\frac12 m(m+1)+1 = \frac12 m(2n-m+1)+1$ parameters since $q_{ij}=q_{ji}$. In the scale-invariant case the number is reduced to $m(n-m)$ and there is a natural duality with respect to interchange of function values and derivatives in (\ref{limit:mc}).

\smallskip

\noindent (d) On the other hand, if $\lambda$ and $q_{ij}$ are nonzero, the operator $H$ may have a discrete spectrum in $(-\infty,0)$. Since any such operator and the Dirichlet decoupled one have a common symmetric restriction with deficiency indices not exceeding $(n,n)$, it follows from general principles \cite[Sect.~8.3]{Wei80} that the number of such eigenvalues does not exceed $n$ counting multiplicities.
The actual number depends on the parameter values. For example, if $m=1$ and all the $\theta_{1j}$ and $q_{1j}$ are the same, the conditions (\ref{limit:mc}) are equivalent the usual $\delta$-coupling \cite{Ex96} which has one or no eigenvalue depending on the sign of $\lambda$.
}
\end{rems}

Our first main result says that the Schr\"{o}dinger operators $H_\varepsilon$ approach $H$ as $\varepsilon\to 0$ in the norm-resolvent topology with a particular convergence rate:

\begin{thm}\label{thm:operator}
$H_\varepsilon \to H$ holds as $\varepsilon\to0$ in the norm resolvent sense, and morever, for any fixed $\zeta\in\mathbb{C}\setminus\mathbb{R}$ there is a constant $C$ such that
\[
\|(H_\varepsilon-\zeta)^{-1}-(H-\zeta)^{-1}\|
_{\mathcal{B}({L}^2(\Gamma))}
\leq C\sqrt\varepsilon,\qquad \varepsilon\in(0,1].
\]
\end{thm}

\medskip

The second question to address concerns the scattering. We denote by $S_0$ the Schr\"{o}dinger operator describing a free particle moving on the graph $\Gamma$, i.e.
\[
S_0=-\frac{\mathrm{d}^2}{\mathrm{d}x^2},
\qquad \mathrm{dom}\,S_0={H}^2(\Gamma)
\cap{K}(\Gamma).
\]
We are interested in the asymptotic behavior of the scattering amplitudes with respect to this free dynamics. Our second main result is the following:
\begin{thm}\label{thm:Scat}
For any momentum $k>0$ the on-shell scattering matrix for the pair $(H_\varepsilon,S_0)$ converges as $\varepsilon\to0$ to that of $(H,S_0)$.
\end{thm}


Proofs of the above claims are the contents of the following two sections.
Since we will employ many constants, a comment on them is due. The uppercase $C_1\dots,C_{m+1}$ will appear in the statement of Lemma~\ref{lem:z_eps}, similarly uppercase $B_1,\dots,B_5$ will be used in Lemmata~\ref{lem:jump1}--\ref{lem:inequalities}. On the other hand, the lowercase $c_i$'s shall denote various positive numbers independent of $\varepsilon$ the values of which may be different in different proofs.

\section{The operator convergence}
\setcounter{equation}{0}

This section is devoted to proof of Theorem~\ref{thm:operator}. We will do that by considering the functions $y_\varepsilon:=(H_\varepsilon -\zeta)^{-1}f$ and $y:=(H-\zeta)^{-1}f$ for a fixed $\zeta\in \mathbb{C}\setminus\mathbb{R}$ and any $f\in{L}^2(\Gamma)$ and demonstrating that
\[
\|y_\varepsilon-y\|_{{L}^2(\Gamma)}\leq C\sqrt\varepsilon\,\|f\|_{{L}^2(\Gamma)}
\]
holds with a constant $C$ independent of $\varepsilon$ and $f$. To this aim, we are going to construct a function $\tilde y_\varepsilon$
which will be a good approximation to both the $y_\varepsilon$ and $y$.
Since the differential expressions of $H_\varepsilon$ and $H$ coincide on $\Gamma\setminus\Omega_\varepsilon$, it is natural to identify $\tilde y_\varepsilon$ on $\Gamma\setminus\Omega_\varepsilon$ with $y$.
A more subtle construction is needed on $\Omega_\varepsilon$.
To define $\tilde y_\varepsilon$ on $\Omega_\varepsilon$ we have to employ the resonant solutions $\psi_1\dots,\psi_m$, i.e. the nontrivial solutions of the problem (\ref{Resonant}), in combination with a corrector function $z_\varepsilon$ being a solution of a particular nonhomogenous problem on $\Omega$, namely
\begin{eqnarray} 
-z''+Qz=\varepsilon f(\varepsilon\cdot)-{\lambda}Q
\sum_{{i}\in\mathfrak{m}}y(a_{i}^\varepsilon)\psi_{i}
\quad\mathrm{on}\quad\Omega, \nonumber \\[-.7em]
\label{u_eps} \\[-.7em]
\frac{\mathrm{d}z}{\mathrm{d}\omega_{i}}(a_{i})=
\nu_{i},\quad {i}\in\frak{m},\qquad
\frac{\mathrm{d}z}{\mathrm{d}\omega_i}(a_i)=
y'(a_i^\varepsilon),
\quad i\in \mathfrak{n}\setminus\mathfrak{m}, \nonumber
\end{eqnarray} 
which obey the Kirchhoff conditions at the vertex $a$ and  $z_\varepsilon (a_{i})=0$ for ${i}\in\mathfrak{m}$. The problem (\ref{u_eps}) admits a solution if and only if
\begin{equation} 
\label{kappa_m}
\hspace{-2em} \nu_{i}=-\Bigg[\sum_{{j}\in\mathfrak{n}\setminus \mathfrak{m}}y'(a_{j}^\varepsilon)\theta_{{i}{j}}
-{\lambda}\sum_{j\in\mathfrak{m}}y(a_j^\varepsilon) {q}_{ij}
+\varepsilon\int_\Omega f(\varepsilon t)\psi_{i}(t)\,\mathrm{d}\Omega\Bigg],
\quad {i}\in\mathfrak{m},
\end{equation} 
cf.~\cite[Thm~XI.4.1]{CodLev55} for details in the similar one-dimensional case; the value of $\nu_{i}$ is obtained by multiplying equation (\ref{u_eps}) by the function $\psi_{i}$ and integration by parts.
Using the variation-of-constants method, every solution of the nonhomogenous problem can be written as $\tilde z_\varepsilon+\sum_{{i}\in\mathfrak{m}} c_{i}\psi_{i}$, where $\tilde z_\varepsilon$ is a fixed solution and $c_{i}$ are uniquely determined coefficients, thus the corrector $z_\varepsilon$ is of the form $z_\varepsilon=\tilde z_\varepsilon-\sum_{{i}\in\mathfrak{m}} \tilde z_\varepsilon(a_{i}) \psi_{i}$. We can make the following claim:

\begin{lem}\label{lem:z_eps}
For any $f\in {L}^2(\Gamma)$ and $\varepsilon\in(0,1]$ we have the inequalities
\begin{eqnarray*}
\bigg|\nu_{i}-\frac{\mathrm{d}y}{\mathrm{d}\gamma_{i}}(a)\bigg|
\leq C_{i}\sqrt\varepsilon\,\|f\|_{{L}^2(\Gamma)},
\quad {i}\in\mathfrak{m},
\\ \quad\;\;\;
\|z_\varepsilon\|_{{H}^2(\Omega)}\leq C_{m+1}\|f\|_{{L}^2(\Gamma)}.
\end{eqnarray*}
\end{lem}
\begin{proof}
First we observe that the resolvent $(H-\zeta)^{-1}$ is a bounded operator from ${L}^2(\Gamma)$ to the domain of $H$ equipped with the graph norm. Since the latter space is a subspace of ${H}^2(\Gamma)$, it follows that
\begin{equation}\label{est:y0}
\|y\|_{{H}^2(\Gamma)}\leq c_1\|f\|_{{L}^2(\Gamma)},
\end{equation}
and consequently,
\begin{equation}\label{est:y}
\|y\|_{{BC}^1(\Gamma)}\leq c_2\|f\|_{{L}^2(\Gamma)}
\end{equation}
in view of the fact that ${H}^2(\Gamma)\subset{BC}^1(\Gamma)$ by the Sobolev embedding theorem. Subtracting the relation
\[
\frac{\mathrm{d}y}{\mathrm{d}\gamma_{i}}(a)=
-\sum_{{j}\in\mathfrak{n}\setminus \mathfrak{m}}\theta_{{i}{j}}\frac{\mathrm{d}y}
{\mathrm{d}\gamma_{j}}(a)
+
{\lambda}
\sum_{j\in\mathfrak{m}}{q}_{ij}y_{\gamma_j}(a)
\]
which is a part of (\ref{limit:mc}) from formula~(\ref{kappa_m}) we arrive at
\begin{eqnarray} 
\lefteqn{\bigg|
\nu_{i}-\frac{\mathrm{d}y}{\mathrm{d}\gamma_{i}}(a)\bigg|
\leq\sum_{{j}\in\mathfrak{n}\setminus\mathfrak{m}}
|\theta_{{i}{j}}|
\bigg|y'(a_{j}^\varepsilon)-
\frac{\mathrm{d}y}{\mathrm{d}\gamma_{j}}(a)\bigg|} \nonumber
\\
\phantom{=}\,
+|{\lambda}|
\sum_{j\in\mathfrak{m}}|{q}_{ij}|
|y(a_j^\varepsilon)-y_{\gamma_j}(a)| \label{est:kappam}
\\
\phantom{=}\,
+\varepsilon \|f(\varepsilon\cdot)\|_{{L}^2(\Omega)}
\|\psi_{i}\|_{{L}^2(\Omega)}
\leq C_{i}
\sqrt\varepsilon\,\|f\|_{{L}^2(\Gamma)}, \nonumber
\end{eqnarray} 
proving thus the first inequality. We have used here the estimates
\begin{eqnarray} 
\big|y(a_i^\varepsilon)-
y_{\gamma_i}(a)\big|
\leq
\int_a^{a_i^\varepsilon}|y'|\,\mathrm{d}\gamma_i
\leq c_3\sqrt\varepsilon\,\|y\|_{{H}^2(\Gamma)}\leq
c_4\sqrt\varepsilon\,\|f\|_{{L}^2(\Gamma)}, \nonumber
\\[-.5em] \label{est:1derivative} \\[-.5em]
\bigg|y'(a_i^\varepsilon)-
\frac{\mathrm{d}y}{\mathrm{d}\gamma_i}(a)\bigg|
\leq
\int_a^{a_i^\varepsilon}|y''|\,\mathrm{d}\gamma_i
\leq c_5\sqrt\varepsilon\,\|y\|_{{H}^2(\Gamma)}\leq
c_6\sqrt\varepsilon\,\|f\|_{{L}^2(\Gamma)}, \nonumber
\end{eqnarray} 
which hold for all $i\in\mathfrak{n}$ by virtue of the Cauchy-Bunyakovsky-Schwarz inequality and the fact that there is a constant $c_7$ such that
\begin{equation}\label{f_eps}
\|f(\varepsilon\cdot)\|_{{L}^2(\Omega)}\leq \frac{c_7}{\sqrt\varepsilon}\|f\|_{{L}^2(\Gamma)}.
\end{equation}
The restriction $z_{\varepsilon,\omega_{i}}$ of $z_\varepsilon$ to $\omega_{i}$ solves the initial value problem
\begin{eqnarray*}
-z''+Qz
=\varepsilon f(\varepsilon\cdot)-{\lambda}Q
\sum_{j\in\mathfrak{m}}y(a_j^\varepsilon)\psi_j
\quad\mathrm{on}\quad\omega_{i},\\
\phantom{-}z(a_{i})=0,\qquad
\frac{\mathrm{d}z}{\mathrm{d}\omega_{i}}(a_{i})=\nu_{i}
\end{eqnarray*}
for all ${i}\in\mathfrak{m}$. Using the standard a priori estimate for the solution of the initial value problem --- see, e.g., \cite[Prop~2.3]{Gol13} --- we infer from relations (\ref{est:y}), (\ref{est:kappam}), and (\ref{f_eps}) that
\begin{equation}\label{est:u_epsM}
\|z_\varepsilon\|_{{H}^2(\omega_{i})}\leq c_8\big(|\nu_{i}|+\|y\|_{{BC}(\Gamma)}+\varepsilon\|f(\varepsilon\cdot)\|
_{{L}^2(\Omega)}\big)
\leq c_9\|f\|_{{L}^2(\Gamma)}.
\end{equation}
Next we claim that at most one of the problems
\[
\hspace{-2em} -z''+Qz=0\quad\mathrm{on}\quad\omega_i,\qquad
z(a)=0,\qquad
\frac{\mathrm{d}z}{\mathrm{d}\omega_i}(a_i)=0,
\qquad
i\in\mathfrak{n}\setminus\mathfrak{m},
\]
has a nontrivial solution. We assume the opposite, namely that there are nonzero solutions of at least two of the problems. In such a case one would be able to construct in a straightforward manner a nontrivial solution of problem (\ref{Resonant}) vanishing on each edge $\omega_{i}$ for ${i}\in\mathfrak{m}$, which is however impossible, since this solution should be linearly independent with all resonant solutions $\psi_{i}, \,{i}\in\mathfrak{m}$.

Without loss of generality we may suppose that the above homogeneous problem on $\omega_i$ does not admit a nontrivial solution for $i\in\mathfrak{n}\setminus(\mathfrak{m}\cup\{n\})$, then for each $i\in\mathfrak{n}\setminus(\mathfrak{m}\cup\{n\})$ the corresponding nonhomogenous problem
\begin{eqnarray*}
-z''+z=\varepsilon f(\varepsilon\cdot)-{\lambda}Q
\sum_{j\in\mathfrak{m}}
y(a_j^\varepsilon)\psi_j\quad\mathrm{on}\quad\omega_i,
\\
\phantom{-}
z(a)=z_{\varepsilon,\omega_1(a)}, \qquad
\frac{\mathrm{d}z}{\mathrm{d}\omega_i}(a_i)=
y'(a_i^\varepsilon)
\end{eqnarray*}
has a unique solution, namely $z_{\varepsilon,\omega_i}$. Moreover,
\begin{equation}\label{est:u_epsN}
\hspace{-2em} \|z_\varepsilon\|_{{H}^2(\omega_i)}\leq c_{10}\big(|z_{\varepsilon,\omega_1(a)}|+
\|y\|_{{C}^1(\Gamma)}
+\varepsilon\|f(\varepsilon\cdot)\|_{{L}^2(\Omega)}\big)
\leq c_{11}\|f\|_{{L}^2(\Gamma)}
\end{equation}
by the a priori estimate for $z_{\varepsilon,\omega_i}$, (\ref{est:y}), (\ref{f_eps}), and (\ref{est:u_epsM}). Finally, the initial value problem
\begin{eqnarray*}
-z''+Qz=\varepsilon f(\varepsilon\cdot)-{\lambda}Q
\sum_{i\in\mathfrak{m}}
y(a_i^\varepsilon)\psi_i\quad\mathrm{on}\quad\omega_n,
\\
\phantom{-}
z(a)=z_{\varepsilon,\omega_1(a)}, \qquad
\frac{\mathrm{d}z}{\mathrm{d}\omega_n}(a)=
-\sum_{i\in \mathfrak{n}\setminus\{n\}}
\frac{\mathrm{d}z_\varepsilon}{\mathrm{d}\omega_i}(a)
\end{eqnarray*}
yields the function $z_{\varepsilon,\omega_n}$. By virtue of (\ref{f_eps}), (\ref{est:u_epsM}), and (\ref{est:u_epsN}) we find that
\begin{eqnarray*}
\|z_\varepsilon\|_{{H}^2(\omega_n)}\leq c_{12}\Bigg[|z_{\varepsilon,\omega_1(a)}|+
\sum_{i\in\mathfrak{n}\setminus\{n\}}
\bigg|\frac{\mathrm{d}z_\varepsilon}
{\mathrm{d}\omega_i}
(a)\bigg|
\\
\phantom{\leq}\,
+\|y\|_{{BC}(\Omega)}
+\varepsilon\|f(\varepsilon\cdot)\|_{{L}^2(\Omega)}\Bigg]
\leq c_{13}\|f\|_{{L}^2(\Gamma)};
\end{eqnarray*}
combining this estimate with (\ref{est:u_epsM}) and (\ref{est:u_epsN}), we arrive at the second inequality.
\end{proof}

As we have said we seek the approximation function in the form
\[
\varphi_\varepsilon:= \left\{\begin{array}{lll}
y & \mathrm{on} & \Gamma\setminus\Omega_\varepsilon,\\[.5em]
\sum_{{i}\in\mathfrak{m}}y(a_{i}^\varepsilon)
\psi_{i}(\varepsilon^{-1}\cdot)+\varepsilon z_\varepsilon(\varepsilon^{-1}\cdot) & \mathrm{on} &
\Omega_\varepsilon. \end{array}\right.
\]
Although this function is indeed a good approximation to $y_\varepsilon$ and $y$ in some sense, it does not suit our purposes due to the fact that it is not smooth at the point $a_i^\varepsilon$. We are going to show, however, that the jumps of $\varphi_\varepsilon$ and $\varphi'_\varepsilon$ at these points are small which makes it possible to construct another ``small'' corrector $\psi_\varepsilon$ with the property that $\varphi_\varepsilon+\psi_\varepsilon$ is smooth and is still close to $y_\varepsilon$ and $y$.

To the claim about smallness more precise, we introduce the symbol $[g]_{a_i^\varepsilon}$ for the jump of the function $g$ at the point  $a_i^\varepsilon$. We have the following estimates:

\begin{lem}\label{lem:jump1}
There are numbers $B_1$ and $B_2$ such that for all $i\in\mathfrak{n}$ we have
\[
\big|[\varphi_\varepsilon]_{a_i^\varepsilon}\big|\leq
B_1\sqrt\varepsilon\,\|f\|_{{L}^2(\Gamma)},
\qquad
\big|[\varphi'_\varepsilon]_{a_i^\varepsilon}\big|\leq
B_2\sqrt\varepsilon\,\|f\|_{{L}^2(\Gamma)}.
\]
\end{lem}
\begin{proof}
Using the above expression for $\varphi_\varepsilon$ one finds easily
\begin{eqnarray*}
[\varphi_\varepsilon]_{a_i^\varepsilon}=y(a_i^\varepsilon)-
\Bigg[\sum_{j\in\mathfrak{m}}y(a_j^\varepsilon)
\psi_j(a_i)+\varepsilon z_\varepsilon(a_i)\Bigg],
\\[0em]
[\varphi'_\varepsilon]_{a_i^\varepsilon}=
y'(a_i^\varepsilon)-
\frac{\mathrm{d}z_\varepsilon}
{\mathrm{d}\omega_i}(a_i).
\end{eqnarray*}
A straightforward calculation shows that
$[\varphi_\varepsilon]_{a_{i}^\varepsilon}=0$ holds for ${i}\in\mathfrak{m}$ and that $[\varphi'_\varepsilon]_{a_i^\varepsilon}=0$ holds for $i\in\mathfrak{n}\setminus\mathfrak{m}$. Using Lemma~\ref{lem:z_eps}, (\ref{limit:mc}), and (\ref{est:1derivative}) we derive the following estimates for the jumps of the function,
\begin{eqnarray*}
|[\varphi_\varepsilon]_{a_{j}^\varepsilon}|\leq
|y(a_{j}^\varepsilon)-y_{\gamma_{j}}(a)|
+
\sum_{{i}\in\mathfrak{m}}|\theta_{{i}{j}}|
|y(a_{i}^\varepsilon)-y_{\gamma_{i}}(a)|
\\
\phantom{-}\,
+\varepsilon \|z_\varepsilon\|_{{H}^2(\Omega)}\leq B_1\sqrt\varepsilon\,\|f\|_{{L}^2(\Gamma)},
\quad {j}\in\mathfrak{n}\setminus\mathfrak{m},
\end{eqnarray*}
and its derivative,
\begin{eqnarray*}
\big|[\varphi'_\varepsilon]_{a_{i}^\varepsilon}\big|\leq
\bigg|y'(a_{i}^\varepsilon)-\frac{\mathrm{d}y}{\mathrm{d}
\gamma_{i}}(a)
\bigg|+
\bigg|\nu_{i}-\frac{\mathrm{d}y}{\mathrm{d}\gamma_{i}}(a)\bigg|
\leq B_2\sqrt\varepsilon\,\|f\|_{{L}^2(\Gamma)},\quad
{i}\in\mathfrak{m},
\end{eqnarray*}
which completes the proof of the lemma.
\end{proof}

\begin{lem}\label{lem:jump2}
Suppose that $\varphi\in{H}^2
(\Gamma\setminus\{a_i\}_{i\in\mathfrak{n}})$.
Then there exists a function $\psi$ with the following properties:
\begin{itemize}
  \item [(i)] $\;\psi$ belongs to ${C}^\infty
(\Gamma\setminus\Omega)$
and vanishes on $\Omega$;
  \item [(ii)] $\;\varphi+\psi$ is in ${H}^2(\Gamma)$;
\item[(iii)] there is a $B_3$ such that $\psi$ together with its derivatives satisfy the estimates
\[
\max\limits_{x\in\Gamma\setminus\Omega}|\psi^{({j})}(x)|
\leq B_3\sum_{i\in\mathfrak{n}}
\big(\big|[\varphi]_{a_i}\big|
+\big|[\varphi']_{a_i}\big|\big),\qquad j=0,1,2.
\]
\end{itemize}
\end{lem}
\begin{proof}
On every edge $\gamma_i$ we consider infinitely differentiable and compactly supported functions $g_i$ and $h_i$ with the properties that $g_i(a)=\frac{\mathrm{d}h_i}{\mathrm{d}\gamma_i}(a)=1$ and
$h_i(a)=\frac{\mathrm{d}g_i}{\mathrm{d}\gamma_i}(a)=0$, $i\in\mathfrak{n}$. Next denote by $\tilde g_i$ and $\tilde h_i$ their translations $g_i(x-a_i)$ and $h_i(x-a_i)$ extended by zero to the whole graph $\Gamma$; then we define the function $\psi$ as follows:
\[
\psi:=-
\sum_{i\in\mathfrak{n}}
\big(
[\varphi]_{a_i}\tilde g_i+
[\varphi']_{a_i}\tilde h_i
\big).
\]
One can check by a straightforward computation that this function satisfies all the required properties, which completes the proof.
\end{proof}

Combining Lemmata~\ref{lem:jump1} and \ref{lem:jump2}, we conclude that there exists a function $\psi_\varepsilon\in{C}^\infty(\Gamma\setminus \Omega_\varepsilon)$ which vanishes on $\Omega_\varepsilon$ and satisfies the inequalities
\begin{equation}\label{jumps:psi_eps}
\max\limits_{x\in\Gamma\setminus\Omega_\varepsilon}
|\psi_\varepsilon^{(i)}(x)|\leq c_1\sqrt\varepsilon\,
\|f\|_{{L}^2(\Gamma)}
\end{equation}
for $i=0,1,2$, and moreover, that the sum
\[
\tilde y_\varepsilon:=
\varphi_\varepsilon+\psi_\varepsilon
\]
belongs to $\mathrm{dom}\,H_\varepsilon={H}^2(\Gamma)\cap {K}(\Gamma)$.

Now we are in position to justify the choice of $\tilde y_\varepsilon$.

\begin{lem}\label{lem:inequalities}
There are numbers $B_4$ and $B_5$ such that the following inequalities hold:
\begin{eqnarray*}
\|\tilde y_\varepsilon-y_\varepsilon\|_{{L}^2(\Gamma)}
\leq B_4\sqrt\varepsilon\,\|f\|_{{L}^2(\Gamma)},
\\
\|\tilde y_\varepsilon-y\|_{{L}^2(\Gamma)}
\leq B_5\sqrt\varepsilon\,\|f\|_{{L}^2(\Gamma)}.
\end{eqnarray*}
\end{lem}
\begin{proof}
From the definition of $\tilde y_\varepsilon$, we compute
\[
(H_\varepsilon-\zeta)\tilde y_\varepsilon=
f-\psi_\varepsilon''-\zeta\psi_\varepsilon
\]
on $\Gamma\setminus\Omega_\varepsilon$, while on the scaled star graph $\Omega_\varepsilon$ we have
\begin{eqnarray*}
(H_\varepsilon-\zeta)\tilde y_\varepsilon
=
\varepsilon^{-2}I_1(\varepsilon^{-1}x)
+
\varepsilon^{-1}
I_2(\varepsilon^{-1}x)
+I_3(x).
\end{eqnarray*}
Here we have introduced the symbols
\begin{eqnarray*}
I_1(x):=
\sum_{{i}\in\mathfrak{m}}y(a_{i}^\varepsilon)
\big(
-\psi_{i}''+Q(x)\psi_{i}
\big),
\\
I_2(x):=
-z_\varepsilon''+Q(x)z_\varepsilon
+\lambda Q(x)
\sum_{{i}\in\mathfrak{m}}y(a_{i}^\varepsilon)\psi_{i},\\
I_3(x):=
\big[(\lambda(\varepsilon)-1-\varepsilon\lambda)\varepsilon^{-2}
Q(\varepsilon^{-1}x)-
\zeta\big]
\sum_{{i}\in\mathfrak{m}}y(a_{i}^\varepsilon)
\psi_{i}(\varepsilon^{-1}x)
\\
\phantom{:=}\,
+
\big[(\lambda(\varepsilon)-1)\varepsilon^{-1} Q(\varepsilon^{-1}x)-\varepsilon\zeta\big]
z_\varepsilon(\varepsilon^{-1}x).
\end{eqnarray*}
Since the functions $\psi_i$ and $z_\varepsilon$ solve equations (\ref{Resonant}) and (\ref{u_eps}), respectively, it follows that $I_1(x)=0$ and $I_2(x)=\varepsilon f(\varepsilon x)$. We are thus able to
conclude that
\[
(H_\varepsilon-\zeta)\tilde y_\varepsilon=f+r_\varepsilon,
\]
where
\[
r_\varepsilon= \left\{\begin{array}{lll}
-\psi_\varepsilon''-\zeta\psi_\varepsilon &
\mathrm{on} & \Gamma\setminus\Omega_\varepsilon,\\
\phantom{-} I_3
& \mathrm{on} & \Omega_\varepsilon. \end{array} \right.
\]
Hence the formula
\[
\tilde y_\varepsilon-y_\varepsilon=\tilde y_\varepsilon-(H_\varepsilon-\zeta)^{-1}f=
(H_\varepsilon-\zeta)^{-1}r_\varepsilon
\]
yields the following inequality,
\[
\|\tilde y_\varepsilon-y_\varepsilon\|_{{L}^2(\Gamma)}
\leq
\|(H_\varepsilon-\zeta)^{-1}\|_{\mathcal{B}({L}^2(\Gamma))}
\|r_\varepsilon\|_{{L}^2(\Gamma)}
\leq
|\Im\zeta|^{-1}\|r_\varepsilon\|_{{L}^2(\Gamma)}
,
\]
which in view of Lemma~\ref{lem:z_eps} and (\ref{jumps:psi_eps}) gives the required estimate, since
\begin{eqnarray*}
\|r_\varepsilon\|_{{L}^2(\Gamma)}
\leq c_1 \max\limits_{x\in\Gamma\setminus\Omega_\varepsilon}
(|\psi_\varepsilon(x)|+|\psi_\varepsilon''(x)|)+
c_2
\|z_\varepsilon(\varepsilon^{-1}\cdot)\|_{{L}^2(\Omega_\varepsilon)}
\\
\phantom{=}\,
+c_3
\|y\|_{BC(\Gamma)}\sum_{{i}\in\mathfrak{m}}
\|\psi_{i}(\varepsilon^{-1}\cdot)\|_{{L}^2(\Omega_\varepsilon)}
\leq c_4\sqrt\varepsilon\, \|f\|_{{L}^2(\Gamma)}.
\end{eqnarray*}
In a similar way we find that
\begin{eqnarray*}
\|\tilde y_\varepsilon-y\|_{{L}^2(\Gamma)}
\leq c_5 \max\limits_{x\in\Gamma\setminus\Omega_\varepsilon}
|\psi_\varepsilon(x)|+
c_6\|y\|_{BC(\Gamma)}
\sum_{{i}\in\mathfrak{m}}
\|\psi_{i}(\varepsilon^{-1}\cdot)\|_{{L}^2(\Omega_\varepsilon)}
\\
\phantom{-}\,
+c_7\varepsilon\,
\|z_\varepsilon(\varepsilon^{-1}\cdot)\|_{{L}^2(\Omega_\varepsilon)}
+c_8\sqrt\varepsilon\|y\|_{BC(\Gamma)}
\leq c_9\sqrt\varepsilon\, \|f\|_{{L}^2(\Gamma)},
\end{eqnarray*}
which yields the second one of the sought inequalities, thus the lemma is proved.
\end{proof}

\noindent
{\it Proof of Theorem~\ref{thm:operator}}
follows in a straightforward manner from Lemma~\ref{lem:inequalities}. Indeed,
\begin{eqnarray*}
\|(H_\varepsilon-\zeta)^{-1}f-(H-\zeta)^{-1}f\|
_{{L}^2(\Gamma)}
\leq
\|\tilde y_\varepsilon-y_\varepsilon\|_{{L}^2(\Gamma)}\\
\phantom{=}\,
+
\|\tilde y_\varepsilon-y\|_{{L}^2(\Gamma)}\leq
C\sqrt\varepsilon\,\|f\|_{{L}^2(\Gamma)}.
\end{eqnarray*}

\medskip

To conclude this section let us comment on existence of nontrivial limits. We have pointed out that the presence of a zero-energy resonance is a non-generic situation. We can nevertheless provide a constructive algorithm which allows to achieve presence of
a zero-energy resonance of order $m>1$ by adjusting the potential.

\begin{prop}
The problem (\ref{Resonant}) admits $m>1$ nontrivial solutions if and only if $m+1$ problems among the following ones,
\[
-z''+Qz=0\quad\mathrm{on}\quad\omega_i,\qquad
z(a)=0,\qquad
\frac{\mathrm{d}z}{\mathrm{d}\omega_i}(a_i)=0,
\quad
i\in\mathfrak{n},
\]
have nontrivial solutions.
\end{prop}
\begin{proof}
We may suppose without loss of generality that the first $m+1$ problems admit nontrivial solutions which we denote as $z_1$, $\dots$, $z_{m+1}$. To construct $m$ resonant solutions, we set
\[
\psi_{i}:= \left\{ \begin{array}{lll}
\phantom{-}\frac{\mathrm{d}z_{m+1}}{\mathrm{d}\omega_{m+1}}(a)
z_{i} &\quad \mathrm{on} & \omega_{i},\\[.5em]
-\frac{\mathrm{d}z_{i}}{\mathrm{d}\omega_{i}}(a)
z_{m+1} &\quad \mathrm{on} & \omega_{m+1},\\[.5em]
\phantom{-}0 &\quad \mathrm{on} & \Gamma\setminus(\omega_{i}\cup\omega_{m+1})
\end{array}\right.
\]
for $i\in\mathfrak{m}$; then it is straightforward to check that
the functions $\psi_1,\dots,\psi_m$ are non-zero solutions of the problem (\ref{Resonant}).

Conversely, consider $m$ functions $\psi_1,\dots,\psi_m$ solving the problem (\ref{Resonant}); in view of Lemma~\ref{lem:eigenfunction} they can be chosen in such a way that $\psi_{i}(a_{j})=\delta_{{i}{j}}$, with $\delta_{{i}{j}}$ denoting the Kronecker delta. Define functions $z_{i},\,i\in\mathfrak{m}$, by
\[
z_{i}:=\psi_{{i},\omega_{i}};
\]
it is not difficult to see that $\psi_{1,\omega_{j}}$ is non-vanishing for some ${j}\in\mathfrak{n}\setminus\mathfrak{m}$, thus we can set $z_{m+1}:= \psi_{1,\omega_{j}}$, which concludes the proof.
\end{proof}

\section{Convergence of scattering matrices}
\setcounter{equation}{0}

Since the resolvents of $H$ and $S_0$ differ by a finite-rank operator the corresponding wave operators exist and are complete, this is also true for the pair $(H_\varepsilon,S_0)$ in view of the hypotheses we made about the potential $Q$. Our task here is to compare the corresponding on-shell scattering operators which are, of course, $n\times n$ matrices depending on the momentum $k$. We begin with the limit operator $H$ and find scattering amplitudes $T_{{i}{j}}$, after that we will show how to construct scattering solutions for the operator $H_\varepsilon$ given the fundamental system of solutions to the equation on the graph $\Omega$; it is convenient to choose the latter in such a way that its first $m$ elements converge in ${BC}^1(\Omega)$ to the solutions of problem (\ref{Resonant}). In this way we shall be able to analyze the asymptotic behavior of the corresponding scattering amplitudes $T_{ij}^\varepsilon$ for small $\varepsilon$, in particular, to demonstrate that they converge to the corresponding quantities of the limit operator.

Consider thus the scattering on $\Gamma$ for the pair $(H,S_0)$. We employ the natural parametrization $s\in(0,+\infty)$ on the edges of $\Gamma$, where $s=0$ corresponds to the vertex $a$, and suppose that the incoming monochromatic wave $\mathrm{e}^{-\mathrm{i}ks},\:k>0$, follows the edge $\gamma_i$. The scattering solutions then have the form
\begin{equation}\label{ScatteringSolution}
\psi_i(s,k)
= \left\{ \begin{array}{lll}
T_{ij}\,\mathrm{e}^{\mathrm{i}ks} &\quad \mathrm{on} & \gamma_j,\quad j\in\mathfrak{n}\setminus\{i\},\\[.5em]
\mathrm{e}^{-\mathrm{i}ks}+T_{ii}\,\mathrm{e}^{\mathrm{i}ks} &\quad
\mathrm{on} & \gamma_i, \end{array} \right.
\end{equation}
where $\mathrm{i}$ stands for the imaginary unit. As usual, $T_{ii}$ is the reflection amplitude on the edge $\gamma_i$ and $T_{ij}$ are the transmission amplitudes from $\gamma_i$ to the edge $\gamma_{j}$. Finding of them can be reduced to an algebraic problem. Indeed, substituting the scattering solution $\psi_i$ into the matching conditions (\ref{limit:mc}), we get the following linear system,
\begin{equation}\label{sys:T}
\mathcal{A}\mathbf{x}_i=\mathbf{a}_i,
\end{equation}
for the column vector of unknown coefficients $\mathbf{x}_i=(T_{i1},\dots,T_{in})^\top$, $i\in\mathfrak{n}$. The right-hand side of the system is an $n$-element column vector defined as follows: if $m$ is zero, then all the entries of $\mathbf{a}_i$ but $i$-th are zero and the $i$-th entry is $-1$. In the case when $m$ is positive, $\mathbf{a}_i$ is written as
\[
\hspace{-4em} \mathbf{a}_i= \left\{ \begin{array}{lll}
(\lambda q_{1i},\dots,\lambda q_{i-1i},
\mathrm{i}k+\lambda q_{ii},\lambda q_{i+1i},\dots,
\lambda q_{mi},\theta_{im+1},\dots,\theta_{in})^\top &
\mathrm{if} & i\leq m,
\\[.5em]
(\mathrm{i}k\theta_{1i},\dots,
\mathrm{i}k\theta_{mi},0,\dots,0,-1,0,\dots,0)^\top &
\mathrm{if} & i>m \end{array} \right.
\]
with $-1$ in the $i$-th place. Furthermore, by  $\mathcal{A}$ we denote in (\ref{sys:T}) the $n\times n$ complex matrix of the form
\[
\hspace{-4em}
\left(
\begin{array}{cccccccc}
\mathrm{i}k-{\lambda}{q}_{11}&-{\lambda}{q}_{12}&\dots&
-{\lambda}{q}_{1m}&\mathrm{i}k\theta_{1m+1}
&\mathrm{i}k\theta_{1m+2}&\dots&\mathrm{i}k\theta_{1n}
\\[.5em]
-{\lambda}{q}_{21}&\mathrm{i}k-{\lambda}{q}_{22}&\dots&
-{\lambda}{q}_{2m}&\mathrm{i}k\theta_{2m+1}
&\mathrm{i}k\theta_{2m+2}&\dots&\mathrm{i}k\theta_{2n}
\\[.5em]
\vdots&\vdots&\ddots&\vdots&\vdots&\vdots&\ddots&\vdots
\\[.5em]
-{\lambda}{q}_{m1}&
-{\lambda}{q}_{m2}&
\dots&
\mathrm{i}k-{\lambda}{q}_{mm}&
\mathrm{i}k\theta_{mm+1}&\mathrm{i}k\theta_{mm+2}&\dots
&\mathrm{i}k\theta_{mn}
\\[.5em]
-\theta_{1m+1}&-\theta_{2m+1}&\dots&-\theta_{mm+1}&1&0&\dots&0
\\[.5em]
\vdots&\vdots&\ddots&\vdots&\vdots&\vdots&\ddots&\vdots
\\[.5em]
-\theta_{1n}&-\theta_{2n}&\dots&-\theta_{mn}&0&0&\dots&1
\end{array}
\right)
.\]
In what follows, the symbol $\mathcal{A}_{ij}$ stands for the matrix obtained from $\mathcal{A}$ by replacing its ${j}$-th column by $\mathbf{a}_i$ for any fixed $i,j\in\mathfrak{n}$. Then by Cramer's rule the scattering amplitude $T_{ij}$ can be expressed as
\[
T_{ij}=\frac{\det\mathcal{A}_{ij}}{\det\mathcal{A}}.
\]

A more sophisticated analysis is needed to investigate stationary scattering for the pair $(H_\varepsilon,S_0)$. Consider again the incoming monochromatic wave $\mathrm{e}^{-\mathrm{i}ks}$ approaching the vertex along the edge $\gamma_i$.
The corresponding scattering solution $\psi_i^\varepsilon$ has to solve the problem
\begin{equation}\label{scattering:perturbed}
-\psi''+\frac{\lambda(\varepsilon)}{\varepsilon^2}
Q\Big(\frac{x}{\varepsilon}\Big)\psi=k^2\psi
\quad\mathrm{on}\quad\Gamma,\qquad
y\in{K}(\Gamma),
\end{equation}
and, since the potential $Q$ is by assumption supported by $\Omega$, it has the form (\ref{ScatteringSolution}) on $\Gamma\setminus\Omega_\varepsilon$ with the coefficients $T_{ij}$ being replaced by $T_{ij}^\varepsilon$. Thus to solve the scattering problem for the Hamiltonian $H_\varepsilon$ we need to analyze behavior of the amplitudes $T_{ij}^\varepsilon$ as the scaling parameter $\varepsilon\to 0$.

To this aim, we employ linearly independent functions $\varphi_1^\varepsilon$, $\dots$, $\varphi_n^\varepsilon$ solving the Schr\"{o}dinger equation
\begin{equation}\label{eq:varphi}
-\varphi''+\lambda(\varepsilon)Q\varphi
=\varepsilon^2k^2\varphi
\quad\mathrm{on}\quad\Omega
\end{equation}
and obeying the Kirchhoff matching conditions at the vertex $a$. We can choose these functions in such a way that every $\varphi_i^\varepsilon$ depends smoothly on $\varepsilon$ and converges as $\varepsilon\to 0$ in ${BC}^1(\Omega)$ to the function $\varphi_i$ solving the zero-energy Schr\"{o}dinger equation
\[
-\varphi''+Q\varphi
=0
\quad\mathrm{on}\quad\Omega,
\]
and moreover, that the limit $\varphi_i$ is just the resonant solution $\psi_i$ for any $i\in\mathfrak{m}$. Let us specify some properties of such a system $\{\varphi_i^\varepsilon\}$ which we shall need in the following. First of all, in the limit $\varepsilon\to0$ we have
\begin{equation}\label{conv:deriv0}
\varphi_{j}^\varepsilon(a_{i})\to\delta_{{i}{j}},
\qquad  {i},{j}\in\mathfrak{m},
\end{equation}
by virtue of Lemma~\ref{lem:eigenfunction}, and at the same time,
\begin{equation}\label{conv:deriv1}
\frac{\mathrm{d}\varphi_{j}^\varepsilon}
{\mathrm{d}\omega_{i}}(a_{i})
\to0,
\qquad {j}\in\mathfrak{m},\quad {i}\in\mathfrak{n},
\end{equation}
Finally, if we multiply (\ref{eq:varphi}) by $\psi_{i}$ and integrate by parts using Lemma~\ref{lem:eigenfunction} again, we find that
\begin{equation}\label{eq:Fredholm}
\frac{\mathrm{d}\varphi^\varepsilon_{j}}{\mathrm{d}\omega_{i}}
(a_{i})=
-\sum_{{l}\in\mathfrak{n}\setminus \mathfrak{m}}\theta_{{i}{j}}\frac{\mathrm{d}\varphi^\varepsilon_{j}}
{\mathrm{d}\omega_{l}}(a_{l})
+\varepsilon\lambda{q}^\varepsilon_{{i}{j}}+
\OO(\varepsilon^2),\qquad {i}\in\mathfrak{m},\quad {j}\in\mathfrak{n},
\end{equation}
where
\[
{q}^\varepsilon_{{i}{j}}:=
\int_\Omega Q\psi_{i}\varphi^\varepsilon_{j}\,\mathrm{d}\Omega
\]
converges to ${q}_{{i}{j}}$ for all $i,j\in\mathfrak{m}$ as $\varepsilon\to0$.

Given thus the described fundamental system $\{\varphi_i^\varepsilon\}_{i\in\mathfrak{n}}$, we can conclude that the functions $\{\varphi_i^\varepsilon(\cdot/\varepsilon)\}_{i\in\mathfrak{n}}$ form a fundamental system of solutions for (\ref{scattering:perturbed}) on $\Omega_\varepsilon$, so that the scattering solution $\psi_i^\varepsilon$ can be constructed as a linear combination
\begin{equation}\label{ScatSol:small}
\psi_i^\varepsilon=
\sum_{j\in\mathfrak{n}}
C_{ij}^\varepsilon \varphi_j^\varepsilon\Big(\frac{\cdot}{\varepsilon}\Big)
\quad\mathrm{on}\quad\Omega_\varepsilon
\end{equation}
with unknown coefficients $C_{i1}^\varepsilon$, $\dots$, $C_{in}^\varepsilon$.
Since the scattering solution should be continuous at all the points $a_j^\varepsilon$ along with its first derivative, we find from relations (\ref{ScatteringSolution}) and (\ref{ScatSol:small}) that the column vector $\mathbf{x}_i^\varepsilon:= (T^\varepsilon_{i1},\dots, T^\varepsilon_{in},C^\varepsilon_{i1},\dots,C^\varepsilon_{in})^\top$ solves the linear system
\[
\mathcal{A}^\varepsilon\mathbf{x}_i^\varepsilon=
\mathbf{a}_i^\varepsilon,
\]
in which the right-hand side equals
\[
\mathbf{a}_i^\varepsilon=(0,\dots,0,e^{-\mathrm{i}k\varepsilon},
-\mathrm{i}k\varepsilon e^{-\mathrm{i}k\varepsilon},0,\dots,0)^\top
\]
with the $(2i-1)$-th and $2i$-th entries being nonzero, and the $2n\times2n$ matrix $\mathcal{A}^\varepsilon$ is of the form
\[
\hspace{-4em}
\left(
\begin{array}{cccccccc}
-\mathrm{e}^{\mathrm{i}k\varepsilon}&0&0&\dots&0&\varphi_1^\varepsilon(a_1)&\dots
&\varphi_n^\varepsilon(a_1)\\[.5em]
-\mathrm{i}k\varepsilon \mathrm{e}^{\mathrm{i}k\varepsilon}&0&0&\dots&0
&\frac{\mathrm{d}\varphi_1^\varepsilon}{\mathrm{d}\omega_1}(a_1)
&\dots
&\frac{\mathrm{d}\varphi_n^\varepsilon}{\mathrm{d}\omega_1}(a_1)
\\[.5em]
0&-\mathrm{e}^{\mathrm{i}k\varepsilon}&0&\dots&0&\varphi_1^\varepsilon(a_2)
&\dots
&\varphi_n^\varepsilon(a_2)
\\[.5em]
0&-\mathrm{i}k\varepsilon \mathrm{e}^{\mathrm{i}k\varepsilon}&0&\dots&0
&\frac{\mathrm{d}\varphi_1^\varepsilon}{\mathrm{d}\omega_2}(a_2)
&\dots
&\frac{\mathrm{d}\varphi_n^\varepsilon}{\mathrm{d}\omega_2}(a_2)
\\[.5em]
\vdots&\vdots&\vdots&\ddots&\vdots&\vdots&\ddots&\vdots
\\[.5em]
0&0&0&\dots&-\mathrm{e}^{\mathrm{i}k\varepsilon}&\varphi_1^\varepsilon(a_n)&\dots
&\varphi_n^\varepsilon(a_n)
\\[.5em]
0&0&0&\dots&-\mathrm{i}k\varepsilon \mathrm{e}^{\mathrm{i}k\varepsilon}
&\frac{\mathrm{d}\varphi_1^\varepsilon}{\mathrm{d}\omega_n}(a_n)
&\dots
&\frac{\mathrm{d}\varphi_n^\varepsilon}{\mathrm{d}\omega_n}(a_n)
\end{array}
\right)
.\]
Given the matrix $\mathcal{A}^\varepsilon$ we can construct matrices $\mathcal{A}_{ij}^\varepsilon$ with $i,j\in\mathfrak{n}$ in the same way as in the limit operator case, namely by replacing $j$-th column of $\mathcal{A}^\varepsilon$ by $\mathbf{a}_i^\varepsilon$. Using Cramer's rule, we obtain thus the following representation for the scattering amplitudes
\[
T^\varepsilon_{ij}=\frac{\det\mathcal{A}^\varepsilon_{ij}}
{\det\mathcal{A}^\varepsilon},
\quad i,j\in\mathfrak{n}.
\]

Now we are in a position to analyze small $\varepsilon$ behavior of the above determinants:

\begin{lem}
In the limit $\varepsilon\to0$ the determinants of the matrices $\mathcal{A}^\varepsilon$ and $\mathcal{A}_{ij}^\varepsilon$, $i,j\in\mathfrak{n}$, exhibit the following asymptotic behavior:
\[
\det\mathcal{A}^\varepsilon=\varepsilon^m\rho\det\mathcal{A}\,\big(1+o(1)\big),
\qquad
\det\mathcal{A}_{ij}^\varepsilon=
\varepsilon^m\rho\det\mathcal{A}_{ij}\,\big(1+o(1)\big),
\]
where
\[
\rho:=(-1)^{n(n+1)/2+m}
\left|
\begin{array}{ccc}
\frac{\mathrm{d}\varphi_{m+1}}{\mathrm{d}\omega_{m+1}}(a_{m+1})
&
\dots
&
\frac{\mathrm{d}\varphi_{n}}{\mathrm{d}\omega_{m+1}}(a_{m+1})
\\[.5em]
\vdots&\ddots&\vdots
\\[.5em]
\frac{\mathrm{d}\varphi_{m+1}}{\mathrm{d}\omega_n}(a_n)
&
\dots
&
\frac{\mathrm{d}\varphi_{n}}{\mathrm{d}\omega_n}(a_n)
\end{array}
\right|.
\]
\end{lem}
\begin{proof}
We will discuss the behavior of $\det\mathcal{A}^\varepsilon$ only since the corresponding proofs for $\det\mathcal{A}^\varepsilon_{ij}$ proceed in the analogous way. We note that every element of $\mathcal{A}^\varepsilon$ with the indices $(2i,n+j)$ is of the form $\frac{\mathrm{d}\varphi_j^\varepsilon}{\mathrm{d}\omega_i}(a_i)$, which can be in view of (\ref{eq:Fredholm}) written as
\[
-\sum_{l\in\mathfrak{n}\setminus \mathfrak{m}}\theta_{ij}\frac{\mathrm{d}\varphi^\varepsilon_j}
{\mathrm{d}\omega_l}(a_l)
+\varepsilon{\lambda}{q}^\varepsilon_{ij}+
\OO(\varepsilon^2)
\]
for any $i\in\mathfrak{m}$ and $j\in\mathfrak{n}$. Fix $i\in\mathfrak{m}$ and add to the $2i$-th row the $2j$-th row multiplied by $\theta_{ij}$, $j\in\mathfrak{n}$. In this way $\det\mathcal{A}^\varepsilon$ is equal to
\[
\hspace{-6em}
\fontsize{10pt}{1pt}
\left|
\begin{array}{ccccccccc}
-\mathrm{e}^{\mathrm{i}k\varepsilon}&\dots&0&0&\dots&0&
\varphi_1^\varepsilon(a_1)&\dots
&\varphi_n^\varepsilon(a_1)
\\[.5em]
-\mathrm{i}k\varepsilon \mathrm{e}^{\mathrm{i}k\varepsilon}&\dots&0
&-\mathrm{i}k\varepsilon \mathrm{e}^{\mathrm{i}k\varepsilon}\theta_{1m+1}&\dots&
-\mathrm{i}k\varepsilon \mathrm{e}^{\mathrm{i}k\varepsilon}\theta_{1n}
&
\varepsilon{\lambda}{q}_{11}^\varepsilon+\OO(\varepsilon^2)
&\dots
&
\varepsilon{\lambda}{q}_{n1}^\varepsilon+\OO(\varepsilon^2)
\\[.5em]
\vdots&\ddots&\vdots&\vdots&\ddots&\vdots&\vdots&\ddots&\vdots
\\[.5em]
0&\dots&-\mathrm{e}^{\mathrm{i}k\varepsilon}&
0&\dots&0&
\varphi_1^\varepsilon(a_m)&\dots
&\varphi_n^\varepsilon(a_m)
\\[.5em]
0&\dots&-\mathrm{i}k\varepsilon \mathrm{e}^{\mathrm{i}k\varepsilon}
&
-\mathrm{i}k\varepsilon \mathrm{e}^{\mathrm{i}k\varepsilon}\theta_{mm+1}&\dots&
-\mathrm{i}k\varepsilon \mathrm{e}^{\mathrm{i}k\varepsilon}\theta_{mn}
&
\varepsilon{\lambda}{q}_{1m}^\varepsilon+\OO(\varepsilon^2)
&\dots
&
\varepsilon{\lambda}{q}_{nm}^\varepsilon+\OO(\varepsilon^2)
\\[.5em]
0&\dots&0&
-\mathrm{e}^{\mathrm{i}k\varepsilon}&\dots&0&
\varphi_1^\varepsilon(a_{m+1})&\dots
&\varphi_n^\varepsilon(a_{m+1})
\\[.5em]
0&\dots&0&
-\mathrm{i}k\varepsilon \mathrm{e}^{\mathrm{i}k\varepsilon}&\dots&0
&\frac{\mathrm{d}\varphi_1^\varepsilon}{\mathrm{d}\omega_{m+1}}(a_{m+1})
&\dots
&\frac{\mathrm{d}\varphi_n^\varepsilon}{\mathrm{d}\omega_{m+1}}(a_{m+1})
\\[.5em]
\vdots&\ddots&\vdots&\vdots&\ddots&\vdots&\vdots&\ddots&\dots
\\[.5em]
0&\dots&0&\dots&0&-\mathrm{e}^{\mathrm{i}k\varepsilon}
&\varphi_1^\varepsilon(a_n)&\dots
&\varphi_n^\varepsilon(a_n)
\\[.5em]
0&\dots&0&\dots&0&-\mathrm{i}k\varepsilon \mathrm{e}^{\mathrm{i}k\varepsilon}
&\frac{\mathrm{d}\varphi_1^\varepsilon}{\mathrm{d}\omega_n}(a_n)
&\dots
&\frac{\mathrm{d}\varphi_n^\varepsilon}{\mathrm{d}\omega_n}(a_n)
\end{array}
\right|
.\]
For any $i\in\mathfrak{m}$ one can factor $\varepsilon$ out of the $2i$-th row. Next we rearrange the rows in the following way:
first we put the rows with the odd numbers $2i-1$, then we write the ones with the even numbers $2i$, both for $i\in\mathfrak{m}$. Next we put the rows with the odd numbers $2i-1$ and finally with the even numbers $2i$, now for $i\in\mathfrak{n}\setminus\mathfrak{m}$. In this way $\det\mathcal{A}^\varepsilon$ coincides with the determinant
$\Delta^\varepsilon$ defined as
\[
\hspace{-6em}
\fontsize{10pt}{1pt}
\left|
\begin{array}{ccccccccc}
-\mathrm{e}^{\mathrm{i}k\varepsilon}&\dots&0&0&\dots&0
&\varphi_1^\varepsilon(a_1)&\dots
&\varphi_n^\varepsilon(a_1)
\\[.5em]
\vdots&\ddots&\vdots&\vdots&\ddots&\vdots&\vdots&\ddots&\dots
\\[.5em]
0&\dots&-\mathrm{e}^{\mathrm{i}k\varepsilon}&
0&\dots&0&
\varphi_1^\varepsilon(a_m)&\dots
&\varphi_n^\varepsilon(a_m)
\\[.5em]
-\mathrm{i}k \mathrm{e}^{\mathrm{i}k\varepsilon}&\dots&0
&-\mathrm{i}k \mathrm{e}^{\mathrm{i}k\varepsilon}\theta_{1m+1}&\dots&
-\mathrm{i}k \mathrm{e}^{\mathrm{i}k\varepsilon}\theta_{1n}
&
{\lambda}{q}_{11}^\varepsilon
&\dots
&
{\lambda}{q}_{n1}^\varepsilon
\\[.5em]
\vdots&\ddots&\vdots&\vdots&\ddots&\vdots&\vdots&\ddots&\vdots
\\[.5em]
0&\dots&-\mathrm{i}k \mathrm{e}^{\mathrm{i}k\varepsilon}
&
-\mathrm{i}k \mathrm{e}^{\mathrm{i}k\varepsilon}\theta_{mm+1}&\dots&
-\mathrm{i}k\mathrm{e}^{\mathrm{i}k\varepsilon}\theta_{mn}
&
{\lambda}{q}_{1m}^\varepsilon
&\dots
&
{\lambda}{q}_{nm}^\varepsilon
\\[.5em]
0&\dots&0&
-\mathrm{e}^{\mathrm{i}k\varepsilon}&\dots&0&
\varphi_1^\varepsilon(a_{m+1})&\dots
&\varphi_n^\varepsilon(a_{m+1})
\\[.5em]
\vdots&\ddots&\vdots&\vdots&\ddots&\vdots&\vdots&\ddots&\vdots
\\[.5em]
0&\dots&0&\dots&0&-\mathrm{e}^{\mathrm{i}k\varepsilon}
&\varphi_1^\varepsilon(a_n)&\dots
&\varphi_n^\varepsilon(a_n)
\\[.5em]
0&\dots&0&
-\mathrm{i}k\varepsilon \mathrm{e}^{\mathrm{i}k\varepsilon}&\dots&0
&\frac{\mathrm{d}\varphi_1^\varepsilon}{\mathrm{d}\omega_{m+1}}(a_{m+1})
&\dots
&\frac{\mathrm{d}\varphi_n^\varepsilon}{\mathrm{d}\omega_{m+1}}(a_{m+1})
\\[.5em]
\vdots&\ddots&\vdots&\vdots&\ddots&\vdots&\vdots&\ddots&\vdots
\\[.5em]
0&\dots&0&\dots&0&-\mathrm{i}k\varepsilon \mathrm{e}^{\mathrm{i}k\varepsilon}
&\frac{\mathrm{d}\varphi_1^\varepsilon}{\mathrm{d}\omega_n}(a_n)
&\dots
&\frac{\mathrm{d}\varphi_n^\varepsilon}{\mathrm{d}\omega_n}(a_n)
\end{array}
\right|
,
\]
multiplied by $(-1)^{[(m-1)m+(n-m-1)(n-m)]/2}\varepsilon^m$. Adding to the $(m+i)$-th row the $i$-th row multiplied by $-\mathrm{i}k$, we get the following formula for the determinant $\Delta^\varepsilon$:
\[
\hspace{-6em}
\fontsize{10pt}{1pt}
\left|
\begin{array}{ccccccccc}
-\mathrm{e}^{\mathrm{i}k\varepsilon}&\dots&0&0&\dots&0
&\varphi_1^\varepsilon(a_1)&\dots
&\varphi_n^\varepsilon(a_1)
\\[.5em]
\vdots&\ddots&\vdots&\vdots&\ddots&\vdots&\vdots&\ddots&\vdots
\\[.5em]
0&\dots&-\mathrm{e}^{\mathrm{i}k\varepsilon}&
0&\dots&0&
\varphi_1^\varepsilon(a_m)&\dots
&\varphi_n^\varepsilon(a_m)
\\[.5em]
0&\dots&0
&-\mathrm{i}k \mathrm{e}^{\mathrm{i}k\varepsilon}\theta_{1m+1}&\dots&
-\mathrm{i}k \mathrm{e}^{\mathrm{i}k\varepsilon}\theta_{1n}
&
{\lambda}{q}_{11}^\varepsilon-\mathrm{i}k\varphi^\varepsilon_1(a_1)
&\dots
&
{\lambda}{q}_{n1}^\varepsilon-\mathrm{i}k\varphi^\varepsilon_n(a_1)
\\[.5em]
\vdots&\ddots&\vdots&\vdots&\ddots&\vdots&\vdots&\ddots&\vdots
\\[.5em]
0&\dots&0
&
-\mathrm{i}k \mathrm{e}^{\mathrm{i}k\varepsilon}\theta_{mm+1}&\dots&
-\mathrm{i}k\mathrm{e}^{\mathrm{i}k\varepsilon}\theta_{mn}
&
{\lambda}{q}_{1m}^\varepsilon-\mathrm{i}k
\varphi^\varepsilon_1(a_m)
&\dots
&
{\lambda}{q}_{nm}^\varepsilon-\mathrm{i}k
\varphi^\varepsilon_n(a_m)
\\[.5em]
0&\dots&0&
-\mathrm{e}^{\mathrm{i}k\varepsilon}&\dots&0&
\varphi_1^\varepsilon(a_{m+1})&\dots
&\varphi_n^\varepsilon(a_{m+1})
\\[.5em]
\vdots&\ddots&\vdots&\vdots&\ddots&\vdots&\vdots&\ddots&\vdots
\\[.5em]
0&\dots&0&\dots&0&-\mathrm{e}^{\mathrm{i}k\varepsilon}
&\varphi_1^\varepsilon(a_n)&\dots
&\varphi_n^\varepsilon(a_n)
\\[.5em]
0&\dots&0&
-\mathrm{i}k\varepsilon \mathrm{e}^{\mathrm{i}k\varepsilon}&\dots&0
&\frac{\mathrm{d}\varphi_1^\varepsilon}{\mathrm{d}\omega_{m+1}}(a_{m+1})
&\dots
&\frac{\mathrm{d}\varphi_n^\varepsilon}{\mathrm{d}\omega_{m+1}}(a_{m+1})
\\[.5em]
\vdots&\ddots&\vdots&\vdots&\ddots&\vdots&\vdots&\ddots&\vdots
\\[.5em]
0&\dots&0&\dots&0&-\mathrm{i}k\varepsilon \mathrm{e}^{\mathrm{i}k\varepsilon}
&\frac{\mathrm{d}\varphi_1^\varepsilon}{\mathrm{d}\omega_n}(a_n)
&\dots
&\frac{\mathrm{d}\varphi_n^\varepsilon}{\mathrm{d}\omega_n}(a_n)
\end{array}
\right|
.\]
In view of (\ref{conv:deriv0}) and (\ref{conv:deriv1}), we find that $\Delta^\varepsilon$
tends to $(-1)^{m+n+m(n-m)}\rho\det\mathcal{A}$ as $\varepsilon\to0$, and consequently,
\[
\det\mathcal{A}^\varepsilon=
\varepsilon^m\rho\det\mathcal{A}\big(1+o(1)\big)
\quad \mathrm{as} \quad \varepsilon\to0\,;
\]
the asymptotic expansion $\det\mathcal{A}_{ij}^\varepsilon=\varepsilon^m\rho \det\mathcal{A}_{ij}\big(1+o(1)\big)$ for any $i,j\in\mathfrak{n}$ is obtained in a similar way.
\end{proof}

\noindent
{\it Theorem~\ref{thm:Scat}} is now a direct consequence of this lemma in combination with the explicit expressions for the scattering amplitudes of the involved operators.

\medskip

We conclude this section with the analysis of the limiting scattering matrix $\mathcal{S}:=\{T_{ij}\}_{i,j\in\mathfrak{n}}$.
We first observe that if $m$ is zero, i.e., if the Schr\"{o}dinger operator $S$ has no zero-energy resonances,
then $\mathcal{S}=-I$ as it should be expected in the situation when the graph $\Gamma$ decomposes into disconnected edges, each being described by the Dirichlet Laplacian on the respective halfline.

Consider next the small $k$ behavior of the scattering matrix $\mathcal{S}$ and define $n\times n$ matrices $\mathcal{B}$ and $\mathcal{B}_{ij}$ as the matrices $\mathcal{A}$ and $\mathcal{A}_{ij}$ introduced above but with $k=0$. In the limit $k\to0$ the scattering matrix $\mathcal{S}$ tends obviously to the matrix whose entries are defined as $\det \mathcal{B}_{ij}/ \det \mathcal{B}$. Since the matrix $\mathcal{B}_{ii}$ differs from $\mathcal{B}$ by the sign of the $i$-th column only, we conclude that the diagonal elements of $\mathcal{S}$ tend to $-1$ as $k\to0$. Furthermore, using the fact that $i$-th and $j$-th columns of $\mathcal{B}_{ij}$ differ by sign for $j\in\mathfrak{n}$ and $i\in\mathfrak{m}\setminus\{j\}$, we infer that $\det \mathcal{B}_{ij}=0$. On the other hand, for $i\in\mathfrak{n}\setminus\mathfrak{m}$ we expand the determinant of $\mathcal{B}_{ij}$ in terms of the ${j}$-th column; this allows us to conclude that $\det \mathcal{B}_{ij}$ can be written as $(-1)^{i+j}$ multiplied by the minor of the element of $\det \mathcal{B}_{ij}$ with the indices $(i,j)$. Since for all ${j}\in\mathfrak{n}\setminus\{i\}$ such minors contain zero column, it follows that $\det \mathcal{B}_{ij}$ is zero again.

What is important, however, is that the above argument works only under the assumption $\lambda\ne 0$ because otherwise the limit of the denominator is zero and one would obtain an indeterminate expression. We see that for a non-vanishing $\lambda$ the scattering matrix tends to $-I$ as $k\to0$, which means in view of Theorem~\ref{thm:Scat} that the potential $\lambda(\varepsilon) \varepsilon^{-2} Q(\varepsilon^{-1}\cdot)$ becomes asymptotically opaque in the low energy limit.

To treat the case left out, $\lambda=0$, we introduce the $n\times n$ and $1\times n$ matrices
\[
\mathcal{C}:=
\left(
\begin{array}{cccccc}
1&\dots&0&\theta_{1m+1}&\dots&\theta_{1n}
\\[.5em]
\vdots&\ddots&\vdots&\vdots&\ddots&\vdots
\\[.5em]
0&\dots&1&\theta_{mm+1}&\dots&\theta_{mn}
\\[.5em]
-\theta_{1m+1}&\dots&-\theta_{mm+1}&1&\dots&0
\\[.5em]
\vdots&\ddots&\vdots&\vdots&\ddots&\vdots
\\[.5em]
-\theta_{1n}&\dots&-\theta_{mn}&0&\dots&1
\end{array}
\right)
,
\]
and
\[
\mathbf{c}_i:=
\left\{
\begin{array}{lll}
(0,\dots,0,
1,0,\dots,
0,\theta_{im+1},\dots,\theta_{in})^\top
&\mathrm{if}& i\leq m,
\\[.5em]
(\theta_{1i},\dots,
\theta_{mi},0,\dots,0,-1,0,\dots,0)^\top
&
\mathrm{if}& i> m,
\end{array}
\right.
\]
respectively. Arguing as above, we use $\mathcal{C}$ and $\mathbf{c}_i$ to construct the matrices $\mathcal{C}_{ij}$ which make it possible to write the scattering matrix elements as $\det\mathcal{C}_{ij}/\det \mathcal{C}$. It is independent of the momentum $k$ which is not surprising; recall that the conditions (\ref{limit:mc}) with $\lambda=0$ do not couple function values and derivatives, and as a result, the corresponding vertex coupling is scale-invariant -- cf.~Remark~\ref{bc-rem}b. Note that $\mathcal{S}$ is in general different from $-I$ as the example of a star graph with Kirchhoff coupling shows \cite{ExSe89}, hence the assumption $\lambda\ne 0$ in the previous paragraph is indeed substantial.

Let us finally comment on the large $k$ behavior of the scattering amplitudes. It is easy to see from (\ref{sys:T}) that in the limit $k\to\infty$ the value of $\lambda$ is not important and the scattering matrix $\mathcal{S}$ tends to the scattering matrix corresponding to the scale-invariant situation. Consequently, the scattering amplitudes for the Schr\"{o}dinger operator $-\frac{\mathrm{d}^2}{\mathrm{d}x^2} +\lambda(\varepsilon)\varepsilon^{-2}Q(\varepsilon^{-1}\cdot)$ coincide asymptotically with that of $-\frac{\mathrm{d}^2}{\mathrm{d}x^2} +\varepsilon^{-2}Q(\varepsilon^{-1}\cdot)$. It means, in particular, that none of the matching conditions (\ref{limit:mc}) is of the type conventionally called $\delta'$ as we have mentioned in the introduction.

\subsection*{Acknowledgments}
The authors are grateful to Rostyslav Hryniv for careful reading of the manuscript
and valuable remarks.
The research was supported by the Czech Science Foundation within
the project P203/11/0701 and by the European Union with the project ``Support for research teams on CTU'' CZ.1.07/2.3.00/30.0034.

\end{document}